\documentclass[aps,preprint,preprintnumbers,nofootinbib]{revtex4}%

\usepackage{xcolor}
\usepackage{amsfonts}
\usepackage{amssymb}
\usepackage{amsmath}
\usepackage{bm}
\usepackage{amsthm}
\usepackage{graphicx}%
\providecommand{\U}[1]{\protect\rule{.1in}{.1in}}
\theoremstyle{plain}
\newtheorem{thm}{Theorem}
\theoremstyle{definition}

\theoremstyle{proposition}

\theoremstyle{lemma}
\newtheorem{lemma}{Lemma}
\theoremstyle{corollary}
\newtheorem{coro}{Corollary}
\begin{document}
\title{On the $C^k$-embedding of Lorentzian manifolds in Ricci-flat spaces}

\begin{abstract}
In this paper we investigate the problem of non-analytic embeddings of
Lorentzian manifolds in Ricci-flat semi-Riemannian spaces. In order to do
this, we first review some relevant results in the area, and then
motivate both the mathematical and physical interest in this problem. We show
that any $n$-dimensional compact Lorentzian manifold $(M^{n},g)$, with $g$ in the Sobolev space $
H_{s+3}$, $s>\frac{n}{2}$, admits an isometric embedding in an $(2n+2)$%
-dimensional Ricci-flat semi-Riemannian manifold. The sharpest result
available for this type of embeddings, in the general setting, comes as a
corollary of Greene's remarkable embedding theorems [R. Greene, Mem. Am. Math. Soc. 97, 1 (1970)], which guarantee the embedding of a compact $n$-dimensional semi-Riemannian manifold
into an $n(n+5)$-dimensional semi-Euclidean space, thereby guaranteeing the
embedding into a Ricci-flat space with the same dimension. The theorem
presented here improves this corollary in $n^{2}+3n-2$ codimensions by
replacing the Riemann-flat condition with the Ricci-flat one from the beginning. Finally, we will present a corollary of this theorem, which shows that a compact strip in an $n$-dimensional globally hyperbolic space-time can be embedded in a $(2n+2)$-dimensional Ricci-flat semi-Riemannian manifold. 

\end{abstract}
\author{R. Avalos$^1$, F. Dahia$^{2}$, C. Romero$^{2}$}
\affiliation{${}^{2}\!$Departamento de F\'{\i}sica, Universidade Federal da Para\'{\i}ba, Caixa
Postal 5008, 58059-970 Jo\~{a}o Pessoa, PB, Brazil.\\
${}^{1}\!$Departamento de Matematica - UFC, Bloco 914 – Campus do Pici, 60455-760 Fortaleza, Cear\'a,  Brazil.\\
E-mail: rodrigo.avalos@fisica.ufpb.br; fdahia@fisica.ufpb.br; cromero@fisica.ufpb.br}
\maketitle

\address{}

\section{Introduction}

The problem of embedding general $n$-dimensional manifolds in higher-dimensional spaces with particular properties has been an active object of study for more than a century. It is a well known fact that, as mathematicians began to study abstract manifolds, the question of whether these structures were actually
more general than the submanifolds of Euclidean spaces naturally arose. In this direction, several very interesting theorems were proved, such as Whitney's embedding theorem \cite{Whitney}, stating that any $n$-dimensional manifold can be embedded in $\mathbb{R}^{2n}$; Jannet-Cartan's theorem, which
shows the existence of local isometric embeddings for $n$-dimensional
Riemannian manifolds in $\frac{n(n+1)}{2}$-dimensional Euclidean space in the
case of analytic metrics \cite{Janet}-\cite{Cartan}; Nash's famous embedding
theorem, which shows the non-analytic and global version of Janet-Cartan's
result \cite{Nash} and Greene's generalization of Nash's theorem to
semi-Riemannian geometry \cite{Greene}, where, for instance, it is shown that
a compact $n$-dimensional semi-Riemannian manifold can be embedded in an
$n(n+5)$-dimensional semi-Euclidean space. 

Directly related with our present object of study, it is also worth mentioning the results
obtained by C. J. S. Clarke in \cite{Clarke}, where it is shown that a
semi-Riemannian manifold with metric of signature ``$s$" can be embedded in a semi-Euclidean space of dimension $d=\frac{n}{2}(3n+12)+1-\frac{s}{2}$, for the compact case, and $\frac{n}{6}(2n^{2}+40)+\frac{5}{2}n^2+2-\frac{s}{2}$ for the non-compact case, and, also, that for the case of globally hyperbolic Lorentzian manifolds, Clarke suggests that the dimension of the Lorentzian embedding space can be shown to be $d=\frac{n}{6}(2n^{2}+37)+\frac{5}{2}n^{2}+2$. This last claim suffered of problem related with the so called \textit{folk problems on smoothability}, which is addressed and explained, for instance, in another relevant recent result by O. M\"{u}ller and M. S\'{a}nchez, where it is shown that any globally hyperbolic $n$-dimensional space-time can be isometrically
embedded in a Minkowski space of $d=N_{0}(n)+1$ dimensions, where $N_{0}(n)$
is the optimal dimension for the case of an embedding of an $n$-dimensional
Riemannian manifold into an $N$-dimensional Euclidean space \cite{M-S} (for a discussion
on the value of $N_{0}(n)$ see \cite{Greene}, \cite{Clarke}, \cite{chinos}).

This type of problem has also become an interesting object of study in the
context of modern physical theories. For instance, the Cauchy problem in
general relativity (GR) is nothing more than an embedding problem, in which we
look for an embedding of a 3-dimensional Riemannian manifold in a
4-dimensional Lorentzian manifold, satisfying the Einstein field equations
\cite{CB}. This very same problem is of interest in the context of
higher-dimensional theories of gravity \cite{Maartens},\cite{RS1}%
,\cite{RS2},\cite{Wesson1},\cite{Wesson2}. In this context, in which
space-time is supposed to have more than 4 dimensions, and the ordinary
space-time of GR is considered as a submanifold embedded in this
\textit{bulk}, natural additional embedding problems arise. For example, in
some of these higher-dimensional models, the bulk is supposed to be
characterized by some geometric property, such as being Ricci-flat or, more
generally, an Einstein space \cite{Maartens},\cite{Wesson2}. Thus, it becomes
a straightforward question whether or not any solution of the 4-dimensional
Einstein field equations can be embedded in this type of higher-dimensional
structures. This issue triggered some research and it was shown that, for the
local and analytic case, the embeddings do exist. These theorems are the
content of the Campbell-Magaard theorem and its extensions \cite{Campbell}%
-\cite{Magaard},\cite{DR1},\cite{DR2},\cite{DR3},\cite{DR4},\cite{ADR}. Applications of some of these theorems to physical situations can be reviewed in \cite{SW}. In this direction, it is also worth mentioning the approach taken in \cite{PDL}, where it is shown how to construct local isometric embeddings of vacuum solutions of the Einstein equations, into either Einstein spaces or solutions of the higher-dimensional Einstein equations sourced by a scalar field. 

Even though some of the above results have been around for a while, there is no reference in the literature, as far as we know, of non-analytic extensions of the Campbell-Magaard theorems. These type of results would be of interest not just as mathematically interesting problems, but also because some properties of vital importance in relativistic theories, namely causality and stability, may
demand non-analytic versions of the Campbell-Magaard theorems in order for
them to be applicable to realistic physical theories (for a discussion on this
matter, see \cite{Anderson}-\cite{DR5}). In particular, regarding stability of the embedding, since the Campbell-Magaard theorem, together with all of its extensions, are proved by means of the Cauchy-Kovalesvkya theorem, there is no guarantee that the embedding is stable against small perturbations on the space-time metric, even near the \textit{brane} representing the original space-time. Regarding causality, what should be made clear, is that we have to admit non-analytic solutions to our field equations if we intend to get interactions propagating from initial data with finite speed. These two very important remarks coming from physics, give us a strong motivation for the study of non-analytic embeddings for space-time.  

Another important point  to take into account, besides the regularity considerations, is that if we want to rigorously support the statement that \textit{ordinary} GR solutions can be embedded in these higher-dimensional space-times, then \textit{global} theorems are required. 

With all the above motivations in mind, the aim of this paper is to present a first result in this direction, proving a new global embedding theorem for compact $n$-dimensional Lorentzian manifolds, with metric in the Sobolev space $H_{s+3}$, $s>\frac{n}{2}$ (which requires metrics at least $C^{3}$, and
allows in particular smooth metrics, and in general $C^{k}$ metrics with $k$ depending on $n$), in Ricci-flat semi-Riemannian spaces, where the dimensionality needed for the embedding manifold is $2n+2$. It should be noted that the aforementioned embedding theorems of semi-Riemannian manifolds in
semi-Euclidean spaces, set an upper-bound for novel results concerning embeddings of Lorenzian manifolds in Ricci-flat semi-Riemannian spaces. In particular, the sharpest result previously known for the compact case, which will be our object of study, is the one presented by Robert E. Greene in \cite{Greene}, which guarantees the much stronger statement that an $n$-dimensional compact semi-Riemannian manifold can be embedded in an $n(n+5)$-dimensional Riemann-flat space. Thus, the existence of an embedding
in an $n(n+5)$-dimensional Ricci-flat space is guaranteed. This means that, using the best result known so far, $(n^{2}+3n-2)$ additional dimensions are needed with respect to the results we will present in this paper. Furthermore, we will show how the main theorem presented here, which, as explained above, applies to the compact (without boundary) case, can be used to prove the existence of embeddings for any arbitrary finite \textit{strip} in a globally hyperbolic space-time, with compact Cauchy surfaces, in a Ricci-flat manifold. Moreover, such embeddings can be constructed so as to be stable under small perturbations with respect to the ``space-time" metric. It is, to say the least, a curious fact that for the usual $4$-dimensional space-time of general relativity the embedding is guaranteed in a $10$-dimensional Ricci-flat space.

\section{The embedding problem.}

As we have already stated, the aim of this paper is to prove an embedding
theorem for Lorentzian manifolds in Ricci-flat semi-Riemannian ones.
Surprisingly enough, there is no much reference in the literature of
embeddings of general manifolds in Ricci-flat spaces besides the Cauchy
problem in GR and the Campbell-Magaard theorem. When trying to adapt any of
these results for the case we want to study, we are faced with serious
difficulties. On the one hand, the Campbell-Magaard theorem strongly relies on
the Cauchy-Kovalesvky theorem, which depends crucially on the analyticity
assumptions for the quantities involved. Thus, it is not well-suited as a
starting point for the non-analytic case, which is our present object of
study. This type of difficulty is not an odd feature of embedding problems. An
analogous situation was presented when trying to generalize the Janet-Cartan
embedding theorem, which culminated with Nash's theorem. On the other hand,
the Cauchy problem in GR strongly depends on the hyperbolic character of the
evolution equations and the elliptic one of the constraint equations (see, for
example, \cite{CB},\cite{CB-G}), which depend on the signatures of the
space-time metric and the induced metric on the space-like slices.
Nevertheless, these results require weak regularity assumptions, and we will,
in fact, use some of the results known for the Cauchy problem in GR to show
our main theorem.

\subsection{The Cauchy problem in GR}

The Cauchy problem in GR consists in the following. Given an initial data set
$(M,\bar{g},K)$ where $M$ is an n-dimensional smooth Riemannian manifold
with metric $\bar{g}$ and $K$ is a symmetric second rank tensor field, a
development of this initial data set is a space-time $(V,g)$, such that there
exists an embedding into $V$ with the following properties: \newline i) The
metric $\bar{g}$ is the pullback of $g$ by the embedding $i:M\mapsto V$,
that is $i^{\ast}g=\bar{g}$; \newline ii) The image by $i$ of $K$ is the second
fundamental form of $i(M)$ as a submanifold of $(V,g)$.

In the Cauchy problem in GR we look for a development of an initial data set
such that the resulting space-time satisfies the Einstein equations. It is
usually assumed that $V=M\times\mathbb{R}$. We will adopt this usual setting.

At this point, to study the Cauchy problem, it is customary to consider an
$(n+1)$-dimensional space-time $(V,g)$ and then make an ``$(n+1)$-splitting"
for the metric $g$. This means that we consider local co-frames where we can
write the metric $g$ in a convenient way, such that we have a ``space-time
splitting". In order to do this, a vector field $\beta$, which is constructed
so as to be tangent to each hypersurface $M\times\{t\}$, is used to define the
following local frame
\begin{align*}
e_{i}  &  =\partial_{i}\;\;;\;i=1,\cdots,n\\
e_{0}  &  =\partial_{t}-\beta
\end{align*}
and its dual coframe
\begin{align*}
\theta^{i}  &  =dx^{i}+\beta^{i}dt\;;\;\;i=1,\cdots,n\\
\theta^{0}  &  =dt
\end{align*}
Then we can write the metric $g$ in the form
\[
g=-N^{2}\theta^0\otimes\theta^0+g_{ij}\theta^{i}\otimes\theta^{j}%
\]
where $N$ is a positive function referred to as the \textit{lapse} function,
while the vector field $\beta$ is called the \textit{shift} vector. In this
adapted frame, the second fundamental form on each $M\times\{t\}$ takes the
form
\begin{equation}
K_{ij}=\frac{1}{2N}(\partial_{t}g_{ij}-(\bar{\nabla}_{i}\beta
_{j}+\bar{\nabla}_{j}\beta_{i})) \label{curvext}%
\end{equation}
where $\bar{\nabla}$ denotes the induced connection compatible with the
induced metric $\bar{g}$.

It has been shown that if an initial data set satisfies a particular system of
constraint equations, plus some low regularity assumptions, then it admits an
Einstenian development in a spacetime $V$ satisfying the vacuum Einstein
equations \cite{CB}. The constraints we are referring to, are the following:
\begin{align}
\bar{R}-|K|_{\bar{g}}^{2}+(tr_{\bar{g}}K)^{2} & =0\\
\bar{\nabla}\cdot K-\bar{\nabla}tr_{\bar{g}}K & =0
\end{align}
where $\bar{R}$ represents the Ricci scalar of $\bar{g}$,
$\ |\cdot|_{\bar{g}}$ denotes the pointwise-tensor norm in the metric
$\bar{g}$, and $\bar{\nabla}\cdot K$ denotes the divergence of $K$.
In coordinates, these equations become:
\begin{align}
\bar{R}-K^{ij}K_{ij}+(K_{l}^{l})^{2} & =0\label{hamit}\\
\bar{g}^{ju}\bar{\nabla}_{u}K_{ij}-\bar{\nabla}_{i}K_{l}^{l} &
=0 \label{momentum}%
\end{align}
These equations are considered on a particular initial hypersurface
$M\cong M\times\{t\}$, for example, in the hypersuface defined by $t=0$.

Equations (\ref{hamit})-(\ref{momentum}) are generally posed as a set of
equations for $\bar{g}$ and $K$. We will instead use their
\textit{thin-sandwich} formulation \cite{Wheeler}, in order to pose them for
$N$ and $\beta$ \cite{Bartnik},\cite{ADRL}. This is achieved by using the explicit relation between $K$ and $\dot{g}\doteq \partial_tg|_{t=0}$, given by (\ref{curvext}). Plugging this expression in the constraint equations (\ref{hamit})-(\ref{momentum}), if $R(g)<0$, it is possible equate the lapse function $N$ from the Hamiltonian constraint (\ref{hamit}), giving the following
\begin{equation}\label{lapse}
N=\sqrt{\frac{(\mathrm{tr}_{g} \gamma)^{2}-|\gamma|^{2}_{g}%
}{-R_{g}}}
\end{equation}
where the tensor $\gamma$ has components
\begin{equation}
\gamma_{ij}=\frac{1}{2}\big(\dot g_{ij}-(\nabla_{i}\beta_{j}+\nabla_{j}%
\beta_{i})\big).
\end{equation}

In this setting the momentum constraint (\ref{momentum}) becomes a non-linear second order operator on the shift vector field, defining an equation of the form $\Phi(\psi,\beta)=0$, where $\psi= (g,\dot{g})$ is regarded as \textit{the freely specifiable part} of the initial data set $(g,K)$. The idea is that this procedure can be reversed, \textit{i.e}, whenever the set of equations $\Phi(\psi,\beta)=0$ are well-posed as equations for the shift, taking (\ref{lapse}) as a definition and using (\ref{curvext}), we get an initial data set $(g,K)$ satisfying the vacuum constraint equations (see \cite{ADRL} for further details).  In this context we will have the following initial data defined in its corresponding functional spaces:
\begin{align*}
(\bar{g},\dot{g})\in E_{1}\doteq H_{s+3}(T^{0}_{2}M)\times H_{s+1}%
(T^{0}_{2}M) \; ; \; s> \frac{n}{2}.
\end{align*}
where, as usual, $H_{s}(T^{p}_{q}M)$ represents the Sobolev space of $(p,q)$-tensor fields with $s$-generalized derivatives in $L^{2}$ (see \cite{CB}), and we are left with the analysis of the non-linear equations $\Phi(\psi,\beta)=0$. 

Using the setting presented above, it has been recently shown that, on any $n$-dimensional smooth compact manifold $M$, $n\geq3$, we can always find a smooth solution of the vacuum constraint equations $(\psi_{0}=(g_{0},\dot{g}_{0}),N_{0},\beta_{0})$, such that for any $\psi=(g,\dot{g})$ ($g$ a Riemannian metric) in a sufficiently small $E_1$-neighbourhood of $\psi_{0}$ there is also a unique solution of the vacuum constraint equations, constructed as dicussed above by means of the thin-sandwich formulation. This is part of the content of Theorem 3 in \cite{ADRL}. From now on, we will typically refer to such solution $(\psi_0,N_0,\beta_0)$ as a \textit{reference solution} for the constraint equations. All this, in turn, guarantees that there is an embedding of the Riemannian manifold $(M^{n},g)$ into a Ricci-flat $(n+1)$-dimensional space-time $V$, where the space-time metric would be at least $C^{2}$ (the more regular the initial data is, the more regular the space-time metric will be) \cite{CB}. This new result will be our main tool in proving a new embedding theorem. The main idea for this proof goes as follows:

Suppose we are given a smooth compact $n$-dimensional Lorentzian manifold $(M,g)$ and we want to embed it in some Ricci-flat manifold. If we can perturb the metric $g$ in such a way that the perturbed metric $g'$ is now a properly Riemannian metric, and furthermore $g'$ is part of an initial data set solving the vacuum constraint equations (\ref{hamit})-(\ref{momentum}), then, using standard results on the Cauchy problem for GR, we could guarantee the existence of an isometric embedding of $(M,g')$ into an $(n+1)$-dimensional space-time $(V,\bar g')$. If we can arrange things in such a way that the perturbation $g'$ lies in a sufficiently small $E_1$-neighbourhood of the reference solution of the constraints used in the main theorem presented in \cite{ADRL} (see Theorem 3 therein), which has been explained above, then, we can guarantee the existence of the previous embedding appealing to this theorem, and we would have embeddings for both $g_0$ and $g'$. Roughly speaking, the more or less obvious thing to do at that point, would be to embed $(M,g)$ into the product of these two embedding Lorentzian manifolds, and use the embeddings for $g_0$ and $g'$ to construct the embedding for $(M,g)$.

\subsection{The Main Results}

Having established the framework in which we will be operating, we will now
make our assumptions explicit. We will always work on compact manifolds (without
boundary), and we will assume that we have a smooth reference
solution of the constraint equations on such manifolds, denoted by $(\psi
_{0}\doteq(g_{0},\dot{g}_{0}),N_{0},\beta_{0})$, such that, in the previously
defined topologies, there is a neighbourhood of $\psi_{0}$ where, for each
element in this neighbourhood, the constraint equations have a unique solution.
The existence of such a reference solution is guaranteed by Theorem 3 in \cite{ADRL}. Also, in order to define the Sobolev spaces, we make use of the smooth Riemannian
metric $g_{0}$. We can now prove the following proposition:

\begin{lemma}\label{perturbation}
Given a compact $n$-dimensional Lorentzian manifold $(M^{n},g)$, with $g\in
H_{s+3}$, it is always possible to find a Riemannian metric $\tilde{g}$ on $M$
such that:\newline i) $\tilde{g}$ is as close of $g_{0}$ as we want.\newline
ii) $g=\lambda(\tilde{g}-g_{0})$\newline for some positive constant $\lambda$.
\end{lemma}

\begin{proof}
On $M$ we have both $g$ and $g_0$ defined. Now let
\begin{align*}
TU\doteq \underset{p\in M}{\sqcup}\{ v\in T_pM \; \slash \; g_0(v,v)=1 \}
\end{align*}
be the unit bundle associated to $g_0$, and let
\begin{align*}
F:TU &\mapsto \mathbb{R}\\
v=(x,v_x)& \mapsto \frac{g_x(v_x,v_x)}{{g_{0}}_x(v_x,v_x)}
\end{align*}
Noticing that $F$ defines a continuous function on $TU$, which, since $M$ is compact, is a compact set, then $F$ attains its minimum. Furthermore, since $g$ is Lorentzian, then this minimum must be a negative number. Thus we get that $\exists$ $\alpha<0$ such that
\begin{align*}
\alpha\leq F(v) \; \forall \; v\in TU.
\end{align*}
Then, $\forall$ $\lambda>0$ satisfying $-\lambda<\alpha$ it holds that:
\begin{align*}
-\lambda\leq \frac{g_x(v_x,v_x)}{{g_{0}}_x(v_x,v_x)} \; \forall \; v\in TU,
\end{align*}
which gives us the following
\begin{align}
0\leq \frac{1}{\lambda}g_x(v_x,v_x)+{g_{0}}_x(v_x,v_x) \; \forall \; v\in TU.
\end{align}
Now define $\tilde{g}\doteq\frac{1}{\lambda}g+g_0
$. The claim is that this metric is positive-definite. In order to see this, simply note that for arbitrary $m\in M$ and $\forall$ $V\in T_mM$, $V\neq0$, the following holds:
\begin{align*}
\tilde{g}_m(V,V)=\Vert V\Vert^2_{g_0}\tilde{g}_m\Big(\frac{V}{\Vert V\Vert_{g_0}},\frac{V}{\Vert V\Vert_{g_0}}\Big)>0
\end{align*}
since $\frac{V}{\Vert V\Vert_{g_0}}\in T_mU$, where $\Vert V\Vert_{g_0}=(g_0(V,V))^{\frac{1}{2}}$. This last relation shows that $\tilde{g}$ is a well-defined Riemannian metric on $M$.
Now, in order to show our first statement, we simply note that $\lambda$ can in fact be taken as large as we want, thus, given $\epsilon>0$, $\exists$ $\lambda>0$ satisfying the previous statements together the following one:
\begin{align}\label{lambda}
\lVert \tilde{g}-g_0 \rVert_{H_{s+3}}=\frac{1}{\lambda}\lVert g \rVert_{H_{s+3}}<\epsilon.
\end{align}
Having established our first claim, we see that the second one is a trivial consequence of the definition of $\tilde{g}$, and thus the proposition holds.
\end{proof}

We now present the main theorem:

\begin{thm}\label{main}
Any $n$-dimensional compact Lorentzian manifold $(M,g)$, with $n\geq3$ and
$g\in H_{s+3}$, $s>\frac{n}{2}$, admits an embedding in a Ricci-flat
$(2n+2)$-dimensional semi-Riemannian manifold with index $n+1$ (that is, with
$n+1$ time-like dimensions).
\end{thm}

\begin{proof}
We will start by appealing to Lemma \ref{perturbation} and writing the Lorentzian metric $g$ in the following way
\begin{align}\label{eq1}
g=\lambda(\tilde{g}-g_0)
\end{align}
where both $g_0$ and $\tilde{g}$ are Riemannian metrics, and $g_0$ is part of a reference solution $(g_0,\dot{g}_0,N_0,\beta_0)$ of the vacuum constraint equation (\ref{hamit})-(\ref{momentum}) on $M$. As we have already stated above, the main theorem presented in \cite{ADRL}, guarantees that under our hypotheses we can always pick such a solution of the constraint equations, with $g_0$ being actually smooth, and guarantee that for initial data $(g,\dot{g})$ in a small enough neighbourhood of $(g_0,\dot{g}_0)$, there is a unique solution $(N,\beta)$ of the constraint equations (\ref{hamit})-(\ref{momentum}). Thus, picking $\lambda$ and $\tilde{g}$ so that $\tilde{g}$ lies in a small enough $H_{s+3}$-neighbourhood of $g_0$, then, given the initial data $(\tilde{g},\dot{g}_0)$, we know that there is a solution of the constraint equations, and, thus, that there are isometric embeddings $\phi_1$ and $\phi_2$, of $(M,\tilde{g})$ and $(M,g_0)$, into the Ricci-flat Lorentzian manifolds $(V\doteq M\times[0,T),h_1)$ and $(V,h_2)$ respectively, for some $T>0$, with $h_1$ and $h_2$ at least $C^2$. We now consider the semi-Riemannian manifold $(V\times V,h)$, where
\begin{align}\label{eq2}
h\doteq \lambda(\pi^{*}{h_1}-\sigma^{*}h_2)
\end{align}
and $\pi$ and $\sigma$ denote the projections of $V\times V$ onto its first and second factors respectively. Thus, the manifold is the product semi-Riemannian manifold resulting from $(V,h_1)$ and $(V,-h_2)$. We claim that the following map is an isometric embedding
\begin{align}\label{eq3}
\begin{split}
\phi:M&\mapsto V\times V\\
m&\mapsto (\phi_1(m),\phi_2(m))
\end{split}
\end{align}
The fact that $\phi$ is an embedding comes from the fact that both $\phi_1$ and $\phi_2$ are embeddings. To check the isometry condition, given $v,w\in T_m M$, we compute the following:
\begin{align*}
\phi^{*}(h)_m(v,w)&=h(d\phi_m(v),d\phi_m(w))\\
&=\lambda(h_1(d\pi_{\phi(m)}\circ d\phi_m(v),d\pi_{\phi(m)}\circ d\phi_m(w))-h_2(d\sigma_{\phi(m)}\circ d\phi_m(v),d\pi_{\phi(m)}\circ d\phi_m(w)))\\
&=\lambda(h_1({d\phi_{1}}_m(v),{d\phi_{1}}_m(w))-h_2({d\phi_{2}}_m(v),{d\phi_{2}}_m(w)))\\
&=\lambda(\phi_1^{*}(h_1)_m(v,w)-\phi_2^{*}(h_2)_m(v,w))\\
&=\lambda(\tilde{g}-g_0)_m(v,w)\\
&=g_m(v,w).
\end{align*}
This last equality shows the isometry condition. As a final step, we have to show that $(V\times V,h)$ is Ricci-flat. But, since the semi-Riemannian product manifold of Ricci-flat spaces, with the usual product structure, is again Ricci-flat, the conclusion follows immediately.
\end{proof}

This theorem possesses quite some intrinsic geometric value. It provides a general embedding result for compact Lorentzian manifolds with smooth metrics into Ricci-flat spaces, and the codimension needed for the embedding space, even though greater than in the Campbell-Magaard theorem, is much lower than the one needed using the result obtained by R. Greene for the much stronger condition of Riemann-flatness on the embedding space. Explicitly, $(n^{2}+3n-2)$-extra dimensions are saved. Furthermore, we can get a few interesting corollaries from the previous theorem.

\begin{coro}\label{stability}
Any given compact $n$-dimensional Lorentzian manifold $(M,g)$, with metric in the Sobolev space $H_{s+3}$, $s> \frac{n}{2}$, admits an isometric embedding in a $2(n+1)$-dimensional Ricci-flat semi-Riemannian manifold $(\tilde{M}^{2n+2},h)$, with index\:$h$ $=n+1$, where both the embedding and $h$ depend continuously on $g$.  
\end{coro} 
\begin{proof}
From the proof of the previous theorem we get that the embedding $\phi:M\mapsto V\times V$, depends on $g$ only through the embedding $\phi_1:(M,\tilde{g})\mapsto (V=M\times[0,T),h_1)$, $T>0$. Now, from basic facts about the Cauchy problem for the vacuum Einstein equations, we know that $\phi$ is just the inclusion map, and that, under our functional hypotheses, $h_1$ is continuous with respect to $\tilde{g}=\frac{1}{\lambda(g)}g+g_0$. Now, from ($\ref{lambda}$), we see that we any $\lambda(g)>\frac{||g||_{H_{s+3}}}{\epsilon}$ makes the procedure work. In particular, we can fix a choice which satisfies this condition and makes $\lambda$ a continuous function on $g$, for instance $\lambda_0(g)\doteq \frac{||g||_{H_{s+3}}}{\epsilon}+1$. With this choice, it is clear that $\tilde{g}_0\doteq\frac{1}{\lambda_0(g)}g+g_0$ depends continuously on $g$, with respect to the $H_{s+3}$ topology. Thus, using the fact that the Cauchy development of vacuum initial data for the Einstein field equations is continuous with respect to the initial data, then $h_1$ is continuous with respect to $g$, which proves our statement.
\end{proof}

The physical interpretation of this corollary would be that the embedding is \textbf{stable} against small perturbations on the space-time metric $g$. Thus, if we think of a compact space-time $(M,g)$ as an embedded submanifold in $(M^{2n+2},h)$, then, we can guarantee that we can do this embedding in such a way that, near the submanifold representing the 4-dimensional space-time, the embedding is stable against small perturbations on the space-time metric $g$.

Now, we will present an important corollary which shows that Theorem \ref{main} can be used to guarantee the embedding of \textit{strips} of globally hyperbolic space-times with compact space-slices. 

\begin{coro}
Suppose we have a strip $(M\times[0,T),g)$, $T>0$, of an $n$-dimensional globally hyperbolic space-time, where $M$ is an $(n-1)$-dimensional closed manifold, and the Lorentzian metric $g$ is smooth. Then, for any $0< T_1< T_2<T$, the closed strip $M\times [T_1,T_2]$ admits an isometric embedding in a $(2n+2)$-dimensional Ricci-flat manifold $(\tilde{M}^{2n+2},h)$. Furthermore, both the embedding and $h$ can be chosen to be continuous with respect to the space-time metric $g$.
\end{coro}
\begin{proof}
Given $T_1$ and $T_2$ satisfying the hypotheses of the corollary, take $T_3$ and $T_4$ such that $0<T_1<T_2<T_3<T_4<T$. Then, consider two bump functions $f_1,f_2:[0,T]\mapsto \mathbb{R}^{+}$. Pick $f_1$ such that $f_1|_{[0,T_2]}\equiv 1$ and $f_1|_{[T_3,T]}\equiv 0$, and pick $f_2$ such that $f_2|_{[0,T_2]}\equiv 0$ and $f_2|_{[T_3,T]}\equiv 1$. Then, define on $M\times[T_2,T)$ the following $(0,2)$-tensor field
\begin{align}
g_2(\cdot,t)\doteq f_2(t)g(\cdot,t-(T_4-T_1)),
\end{align}
and then extend it for $t\in[0,T_2]$ as zero. In order for the previous expression to be well-defined, take $T_3$ and $T_4$ sufficiently close, and $T_3$ sufficiently close to $T_2$. This choice has to be done such that $T_1+T_2>T_4$. Thus, on $M\times(T_2,T)$ $g_2$ defines a Lorentzian metric. Now, define the following Lorentzian metric on $M\times[T_1,T_4]$:
\begin{align}\label{compactification1}
G\doteq f_1g+g_2.
\end{align} 
Notice that the metric $G$ is actually well-defined on the whole strip $M\times[0,T)$, but we are restricting it to the closed strip $M\times [T_1,T_4]$. Also, note that $G(\cdot,t)=g(\cdot,t)$ $\forall$ $T_1\leq t\leq T_2$, and $G(\cdot,t)=g(\cdot,t-(T_4-T_1)$ for $T_3\leq t\leq T_4$. In particular, this gives us that:
\begin{align*}
&\partial^{k}_tG(\cdot,T_1)=\partial^{k}_tG(\cdot,T_4)=\partial^{k}_tg(\cdot,T_1) \;\;\; \forall \;\;\; k\in\mathbb{N}_{0},\\
&\partial^{k}_aG(\cdot,T_1)=\partial^{k}_aG(\cdot,T_4)=\partial^{k}_ag(\cdot,T_1) \;\;\; \forall \;\;\; k\in\mathbb{N}_{0},
\end{align*}
where $\partial_a$ denote derivatives on the coordinates on $M$. The above relations give us that the metric $G$ can be glued smoothly when identifying the slices $M\times\{T_1\}$ with $M\times\{T_4\}$. This gluing gives us a quotient manifold $\Big(M\times[T_1,T_4]\Big)/\sim \;\;\cong M\times S^1$, with an induced smooth Lorentzian metric $G'$, such that the embedding $(M\times[T_1,T_2],g)\hookrightarrow (M\times S^1,G')$ is isometric. We can thus apply Theorem \ref{main} to $(M\times S^1,G')$ and get an isometric embedding
$(M\times S^1,G')\xrightarrow[]{\phi}(\tilde{M}^{2n+2},h)$. The composition of these two embeddings gives us an isometric embedding of $(M\times[T_1,T_2],g)$ in $(\tilde{M}^{2n+2},h)$. 

Finally, the continuous dependence with respect to the metric $g$ is a consequence of Corollary \ref{stability}.  
\end{proof}

\section{Discussion}

In Theorem \ref{main}, we have obtained a non-analytic embedding theorem for compact Lorentzian manifolds in Ricci-flat spaces, where the regularity assumptions are rather weak. In fact, by means of the Sobolev embedding theorems, we get that the minimal regularity needed for the Lorentzian mentric $g$ is $C^3$, since $H_{s+3}$ embeds in $C^3$ for $s>\frac{n}{2}$. In general, in order for $g\in H_{s+3}$, we will need $g$ in some $C^k$ space, with $k$ depending on $n$, and in particular a smooth $g$ always satisfies this condition. As far as we know, this is a new result, and we regard this theorem as a first step in
the search of sharper theorems in this context. 

Even tough the direct application of this theorem to physics would seem to be limited, since the compactness condition on space-time poses strong consequences, such as the existence of a closed time-like curve (see, for instance, \cite{Oneill}), we have been able to extract two important corollaries from Theorem \ref{main}, which give us that any closed strip in an $n$-dimensional globally hyperbolic space-time, with a compact (without boundary) Cauchy surface, can be isometrically embedded in a $2(n+1)$-dimensional Ricci-flat semi-Riemannian manifold. Notice that such strip can be taken to be as large as we like as long as it remains finite. Thus, for instance, any cosmological model with compact space-slices, considered from some arbitrary finite time $T_1>0$ up to some arbitrary large (but finite) time $T_2<\infty$ can be isometrically embedded in a $2(n+1)$-dimensional Ricci-flat manifold. Nevertheless, it is essential to highlight that the lack of control on the number of time-like dimensions of the ambient space represents a serious restriction for physical applications, which leaves the issue of controlling the signature of the ambient space as a physically relevant open problem.

We would also like to point out that, in  our opinion, the main theorem presented here has an intrinsic geometric value, offering a substantial improvement in the number of dimensions needed to get the embedding. In fact, compared with the sharpest results we are aware of for this kind of embedding we are saving as much as $n^{2}+3n-2$ dimension, which, for instance, in the context of physics where
$n=4$, $26$ extra dimensions are saved.

It is also worth pointing out that results such as the Campbell-Magaard theorems, which have been invoked in the context of some higher-dimensional space-time theories, are both local and analytic theorems. By local we mean that they guarantee the existence of the embedding in a neighbourhood of an arbitrary point in space-time, and, in fact, the size of such neighbourhood is not controlled. In this sense, we have shown that for globally hyperbolic space-times with compact Cauchy hypersurfaces, we can prove the existence of embeddings which are \textit{global in space} and \textit{local in time}, where by local in time we mean that as long as we chose a finite time interval, the embedding will exist. Moreover, regrading analyticity, we have dramatically lowered the regularity needed for the space-time metric, and this has enabled us to prove the stability of the embedding. On the other hand, we would like to remark that it would be desirable to control the number of time-like dimensions of the embedding manifold, even though this has been known to be a difficult task in analogous contexts (see \cite{Greene}). Also, we leave as a future research perspective the weakening of the compactness condition.    

Finally, we would like to stress that we do not expect the results obtained in this paper to be sharp, not regarding the number of extra dimensions needed for the embedding, neither regarding the signature of the embedding space. Probably, embeddings of even lower regularity might be attainable, but we do not regard this technical issue as one of the most important points to be pursued in upcoming research. The reason why we expect sharper results to be attainable, is the strategy adopted to get the embedding, which besides its originality, does not seem to be optimal. Conjecturing which would be the sharpest codimension for this type of embeddings does not seem to be sensible at this point. Certainly the main motivation behind this program is trying to give a response to this issue. Even though the Campbell-Maagaard theorem may suggest that we should ideally look for codimension one, and some explicit examples support this idea, this could easily turn out to be completely misleading, since, as previously explained, this theorem is local and depends crucially on the analyticity assumptions. Since the strategy of the proof has to be substantially modified in order to get global and smooth embeddings, it could easily be the case that the general result may demand more than one codimension. The parallelism with history behind Riemann-flat embeddings, going from Janet-Cartan's theorem to Nash's theorem and continuing up to present days, shows that we should be cautious in conjecturing such optimal codimension. It is certainly a very interesting open problem to establish sufficient conditions which enable us to attain this optimal codimension one global $C^k$-Ricci-flat embedding.
\section*{Acknowledgements}

\noindent R. A. and C. R. would like to thank CNPq and CLAF for financial
support. R. A. would also like to thank CAPES for partial financial support.
We thank the referee for valuable comments and suggestions.

\end{document}